%% file: Polar_ML_BEC.tex
\documentclass[conference,letterpaper]{IEEEtran}

\usepackage[utf8]{inputenc} 
\usepackage[T1]{fontenc}
\usepackage{url}
\usepackage{ifthen}
\usepackage{cite}
\usepackage[cmex10]{amsmath} 
\usepackage{subcaption}
\usepackage{amsthm}
\usepackage{amsfonts}
\usepackage{mathtools}
\usepackage{enumitem}

\newcommand{\bbA}{\mathbb{A}}

\input{shortcuts.tex}

\usepackage{color}
\definecolor{alizarin}{rgb}{0.82, 0.1, 0.26}

\mathtoolsset{showonlyrefs=true}
\begin{document}
\title{Efficient Maximum Likelihood Decoding of Polar Codes Over the Binary Erasure Channel}

\author{%
  \IEEEauthorblockN{Yonatan Urman and David Burshtein}\\
  \IEEEauthorblockA{School of Electrical Engineering\\
    Tel-Aviv University\\
    Tel-Aviv 6997801, Israel\\
    Email: yonatanurman@mail.tau.ac.il, burstyn@eng.tau.ac.il}
}

\maketitle

\begin{abstract}
  A new algorithm for efficient exact maximum likelihood decoding of polar codes (which may be CRC augmented), transmitted over the binary erasure channel, is presented.
  The algorithm applies a matrix triangulation process on a sparse polar code parity check matrix,
  followed by solving a small size linear system over GF(2).
  To implement the matrix triangulation, we apply belief propagation decoding type operations.
  We also indicate how this decoder can be implemented in parallel for low latency decoding. Numerical simulations are used to evaluate the performance and computational complexity of the new algorithm.
\end{abstract}

\section{Introduction} \label{sec:intro}
The error rate performance of a polar code \cite{PolarCodes} with a short to moderate blocklength can be significantly improved by concatenating it with a high rate
cyclic redundancy check (CRC) code, and using a CRC-aided successive cancellation (SC) list (SCL) decoder~\cite{SCL}.
However, both the SC and SCL decoders are sequential and thus suffer from high decoding latency and limited throughput.
Improvements to SC and SCL were proposed, e.g., in \cite{alamdar2011simplified,leroux2013semi,sarkis2014fast,li2014low,balatsoukas2015llr,yuan2015low,xiong2015symbol,chen2016reduce,hashemi2018decoder,hashemi2018decoding,giard2018fast,hashemi2019rate}.
An iterative belief propagation (BP) decoder over the polar code factor graph (FG) was proposed in \cite{Arkan2010PolarC, polar_vs_reed}. This decoder is inherently parallel, and allows for efficient, high throughput implementation, optimizations and extensions~\cite{eslami2010on,BP_arc,BP_BEC,Polar_LDPC_conc,bp_early_term,crc_early_term,PCForChannelSrc,PolarBPCRCWarren,BPPermuted,BPL,BPPermutedWarren,yu2019belief,PolarBPCRCBrink}. However, even the CRC aided BP list (BPL) decoder that uses several parallel decoders, one for each permuted FG representation of the polar code, and also incorporates CRC information in the BP decoding, has higher error rate compared to the CRC-aided SCL decoder \cite{PolarBPCRCBrink}.

In this paper, we consider the problem of decoding polar codes (possibly concatenated with a CRC code) over the binary erasure channel (BEC).
We derive a new low complexity algorithm for computing the exact maximum-likelihood (ML) codeword based on inactivation decoding \cite{pishro2004on,LDPC_ML,shokrollahi2006systems,eslami2010on,cocskun2020successive} (see also \cite[Chapter 2.6]{algebraic_coding_theory}).
In \cite{eslami2010on}, it was proposed  to use Algorithm C of \cite{pishro2004on} for improved BP decoding of polar codes. The differences between our work and \cite{eslami2010on} are as follows. 
First, rather than using the large standard polar code FG, we use the method in \cite{SparseGraphsBPPolar} for constructing (offline) a much smaller sparse PCM for polar codes, thus reducing the computational complexity.
In addition, we show how to extend our method to CRC-augmented polar codes, which are important in practice. We also analyze the computational complexity of our method, and indicate how the algorithm can be implemented in parallel.

\section{Background} \label{sec:background}
We use the following notations: We denote by $\cP\left(N, K\right)$ a polar code of blocklength $N$, information size $K$, and rate $K/N$. The number of stages in the polar code FG is $n=\log_2 N$. 
The information and frozen index sets are denoted by $\bbA$ and $\bar{\bbA}$ respectively. We denote by $\tilde{\bu}$, the information bits vector of length $K$, and by $\bu$ the information word, such that $\bu_{\bbA} = \tilde{\bu}$ and $\bu_{\bar{\bbA}} = 0$. That is, $\bu$ is the full input (column) vector, including the frozen bits, and we assume that the frozen bits are set to zero. We denote the $N\times N$ polar generator matrix by $\bG_{N} = \bB_{N}\bF^{\otimes n}$ \cite{PolarCodes}, where $\bB_{N}$ is the bit reversal permutation matrix, and
$
  \bF=\begin{bmatrix}1 & 0\\1 & 1\end{bmatrix}
$.
The codeword is then generated using $\bc^T=\bu^T\bG_{N}$.
The polar code FG is shown in Fig. \ref{fig:PolarFG}.
\begin{figure}
  \centering
  \includegraphics[width=0.6\linewidth]{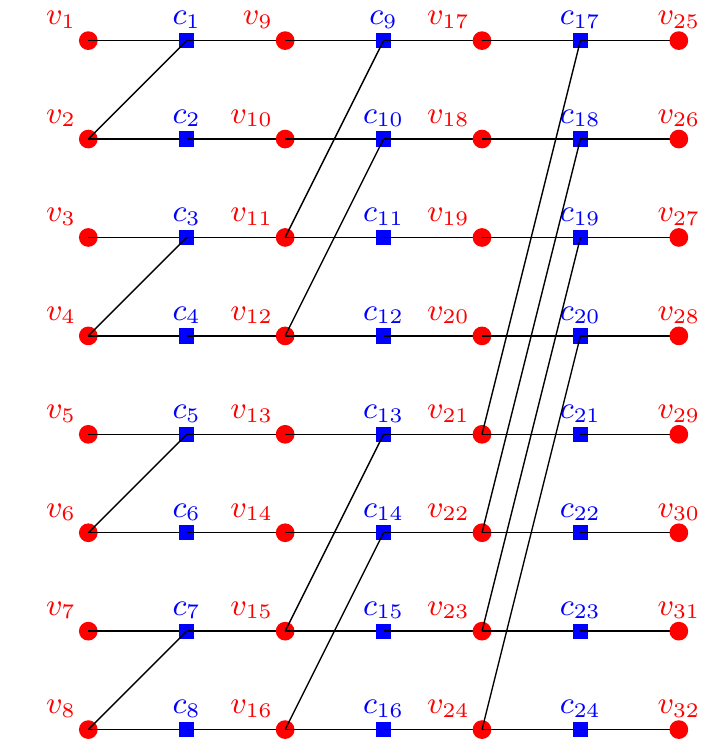}
  \caption{Polar code factor graph for n=3.}
  \label{fig:PolarFG}
\end{figure}
It can be seen that this FG consists of $n$ stages of parity check (PC) nodes, and $n+1$ stages of variable nodes. The variable nodes in the leftmost stage of the graph (denoted in Fig. \ref{fig:PolarFG} by red numbers $1-8$) correspond to the information word $\bu$, and the nodes on the rightmost stage of the graph correspond to the codeword $\bc$ (marked by red numbers $25-32$).
There are three types of variable nodes in the FG:
Channel variable nodes (CVN), corresponding to the codeword,
frozen variable nodes (FVN), and the rest are hidden variable nodes (HVN).

\subsection{Sparse parity check matrix for polar codes \cite{SparseGraphsBPPolar}}
The standard polar code PCM \cite[Lemma 1]{LP_Polar_decoding} is dense. Hence, it is not suitable for standard BP decoding.
For completeness, we briefly review the method in \cite{SparseGraphsBPPolar} for obtaining a sparse representation of the polar code PCM. It starts with the standard polar FG, which can be represented as a PCM of size $N\log_2 N\times N(1+\log_2 N)$, i.e., $N\log_2 N$ PC nodes, and $N(1+\log_2 N)$ variable nodes (out of which, $N-K$ variable nodes are frozen). The leftmost layer of variable nodes of the graph corresponds to columns $1$ to $N$ of the PCM, the next layer of variable nodes corresponds to columns $N+1$ to $2N$, and so on, such that the rightmost variable nodes (codeword) correspond to the last $N$ columns of the PCM. In Fig. \ref{fig:PolarFG}, the column index of variable node $v_i$ is $i$. The resulting PCM is sparse since each PC is connected to at most $3$ variable nodes, as shown in Fig. \ref{fig:PolarFG}. Hence, standard BP can be used effectively on this graph with expected good error probability performance. Unfortunately, the dimensions of the resulting matrix are large (instead of $(N-K)\times N$ for the standard polar code PCM, we now have a matrix of size $N\log_2 N \times N(1+\log_2 N)$ as it contains  HVNs as well), which increases the decoding complexity. To reduce the matrix size, the authors of \cite{SparseGraphsBPPolar} suggested the following pruning steps that yield a valid sparse PCM to the code while reducing the size of the $N\log_2 N\times N(1+\log_2 N)$ original PCM significantly.
\begin{enumerate}[wide, labelwidth=!, labelindent=0pt,label=\textbf{\arabic*})]
  \item \textbf{FVN removal}: If a variable is frozen then it is equal to zero. Thus, all columns that correspond to frozen nodes can be removed from the PCM.
  \item \textbf{Check nodes of degree 1}: The standard polar FG contains check nodes of degrees 2 and 3 only. However, after applying the pruning algorithm, we might get check nodes with degree 1. Their neighboring variable node must be 0. Thus, this check node and its single neighbor variable node can be removed from the FG (the column corresponding to that variable node is removed from the PCM).
  \item \textbf{A CVN connected to a degree 2 check node with an HVN}: The CVN must be equal to the HVN. Thus, the connecting check node can be removed, and the HVN can be replaced (merged) with the CVN.
  \item \textbf{HVN of degree 1}: This HVN does not contribute to the decoding of the other variables since it is $(0,1)$ with probabilities $(1/2,1/2)$. Hence, the variable node and the check node it is connected to can be removed from the graph.
  \item \textbf{HVN of degree 2}: This HVN can be removed and the two connected check nodes can be merged.
  \item \textbf{Degree 2 check node that is connected to two HVNs}: The two HVNs must be equal. Hence, the check node can be removed and the two HVNs can be merged.
\end{enumerate}
Iterating over the above mentioned pruning steps until convergence results in the pruned PCM (FG), which is a valid PCM for the code.
As mentioned before, the last $N$ variable nodes of the resulting PCM correspond to the codeword bits (CVNs), while the rest are HVNs. Standard BP decoding can be used on the new FG.
Note that the pruned PCM has full row rank. This is due to the fact that after the FVN removal step, the dimensions of the PCM are $N\log_2 N \times (N\log_2 N + K)$. Then, for each variable node removed from the graph while executing the other pruning steps summarized above, exactly one check node is removed as well. Hence, at the end of the process, the pruned PCM has dimensions $(N'-K) \times N'$, where $N'\ge N$ is the total number of variable nodes in the pruned graph. Now, $K$ is the dimension of the polar code. Hence the $N'-K$ rows of the pruned PCM must be linearly independent. As an example, for $\cP\left(256, 134\right)$ ($\cP\left(512, 262\right)$, respectively), the algorithm yields a pruned PCM with blocklengh $N'=355$ ($N'=773$) and the fraction of ones in the PCM is $0.7\%$ ($0.33\%$).

\section{Efficient ML Decoding over the BEC} \label{sec:EffML}
Consider the BEC, where a received symbol is either completely known or completely unknown (erased). We denote the erasure probability by $\epsilon$.
The channel capacity is $C(\epsilon) = 1-\epsilon$ \cite{cover_book}.
Denote the input codeword by $\bc$ and the BEC output by $\by$. Since $\bc$ is a codeword it must satisfy $\bH\bc = \mathbf{0}$ where $\bH$ is a PCM of the code.
Denote by $\cK$ the set of known bits in $\bc$ (available from $\by$), and by $\bar{\cK}$ the set of erasures. Denote by $\bH_{\cK}$ ($\bH_{\bar{\cK}}$, respectively) the matrix $\bH$ restricted to columns in $\cK$ ($\bar{\cK}$).
We have, $\mathbf{0} = \bH\bc = \bH_{\cK}\bc_{\cK} + \bH_{\bar{\cK}}\bc_{\bar{\cK}}$. Thus we have the following set of linear equations,
\begin{equation}\label{eq:MLDecBEC}
  \bH_{\bar{\cK}}\bc_{\bar{\cK}} = \bH_{\cK}\bc_{\cK} \: .
\end{equation}
As a result, it can be seen that ML decoding over the BEC is equivalent to solving the set \eqref{eq:MLDecBEC} for $\bc_{\bar{\cK}}$ \cite{modern_coding_theory}. However, the required complexity when using Gaussian elimination is $\mathcal{O}(N^3)$. We now present a much more efficient ML decoding algorithm. This algorithm uses the sparse representation of the polar PCM \cite{SparseGraphsBPPolar} that was reviewed above. This PCM is obtained offline.
Our algorithm is a modified version of the efficient ML decoder \cite{LDPC_ML}, that was proposed for efficient ML decoding of LDPC codes over the BEC.

It may be convenient, for computational efficiency, to store the sparse pruned PCM by the locations of ones at each row and column (along with the number of ones).
It may also be convenient to address the rows and columns using some permutation.
For clarity, we describe the algorithm over the pruned PCM. However, for efficient implementation, some other representation may be preferable (e.g., using the FG for stage 1 described below).
The algorithm has the following stages:
\begin{enumerate}[wide, labelwidth=!, labelindent=0pt,label=\textbf{\arabic*})]
  \item \textbf{Standard BP decoding}. Given the BEC output, apply standard BP decoding on the pruned PCM until convergence.
        If BP decoding was successful (all variable nodes were decoded), we return the decoded word and exit. Otherwise, we proceed to the next decoding stage.
        The PCM after BP decoding is shown in Fig. \ref{fig:PCMAfterBp}.
        In the end of the decoding we have permuted the rows and columns of the PCM such that the first $n_d$ columns correspond to the decoded variable nodes, and the first $n_c$ rows correspond to the decoded PCs (a PC node is said to be decoded if all its neighbor variable nodes have been decoded).
        \begin{figure}
          \centering
          \hspace{-2.7em}
          \includegraphics[width=.75\linewidth]{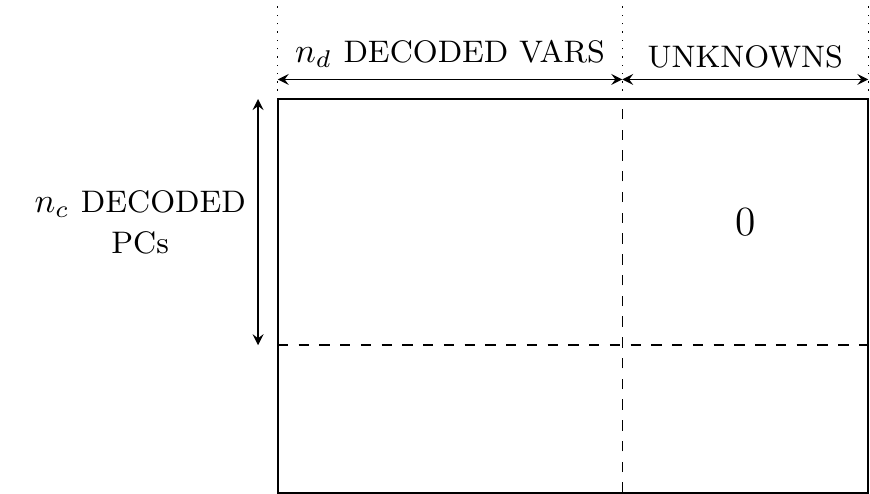}
          \caption{Reordered PCM after initial BP decoding.}
          \label{fig:PCMAfterBp}
        \end{figure}

  \item \textbf{Choosing reference variables and performing triangulation \cite{EffEncLDPC}}. Consider the PCM at the output of the previous stage, shown in Fig. \ref{fig:PCMAfterBp}. Since the BP decoder has converged on this PCM, the number of unknown variable nodes in each undecoded row is larger than one (otherwise the BP could have continued decoding). The goal of the current stage is to bring the PCM to the form shown in Fig. \ref{fig:FinalPCM}, using only row and column permutations, where the $n_c \times n_d$ sub-matrix on the top-left corner is the same as in Fig. \ref{fig:PCMAfterBp}, and where the sub-matrix $\bH^{(1,3)}$ is square lower triangular with ones on its diagonal.
        \begin{figure}
          \centering
          \hspace{-3.7em}
          \includegraphics[width=0.85\linewidth]{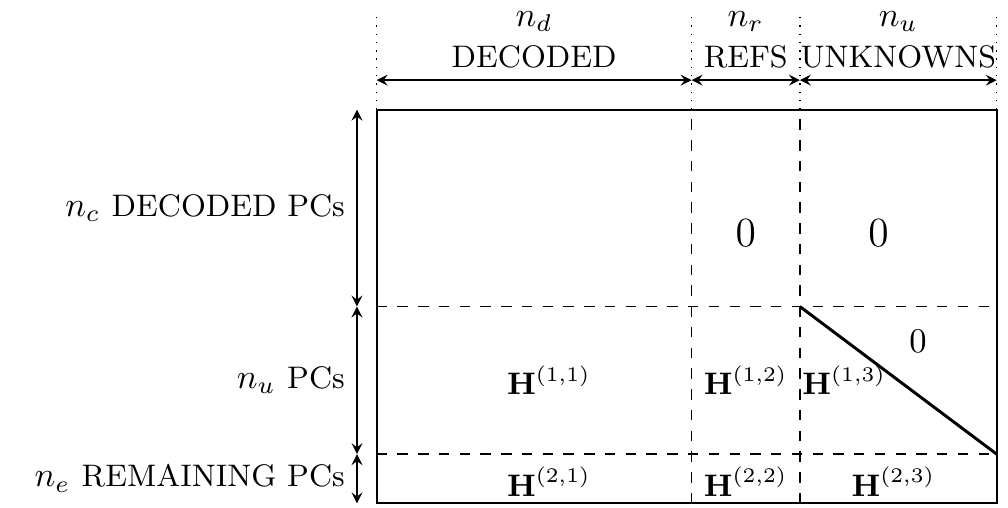}
          \caption{Final PCM. The bold diagonal is filled with ones.}
          \label{fig:FinalPCM}
        \end{figure}

        We start by considering the reordered PCM in Fig. \ref{fig:PCMAfterBp}. We mark $n'_r$ unknown variable nodes (variables that have not been decoded by the BP) as \emph{reference variables} (either by picking them at random or by using a more educated approach as discussed below) and remove them from the list of unknowns. We then permute the matrix columns so that the $n'_r$ columns corresponding to these reference variables are placed immediately after the columns corresponding to the $n_d$ decoded variables. We then perform a \emph{diagonal extension step} \cite{EffEncLDPC} on this column permuted PCM. This means that we check for rows with a single unknown variable node in the remaining columns (those with indices larger than $n_d+n'_r$).
        Assume we found $l_1$ such rows, and the locations of ones in these rows are ${(r_1, c_1), (r_2, c_2), \ldots , (r_{l_1}, c_{l_1})}$. Then, for every $i\in \{1, \ldots, l_1\}$, we permute row $r_{i}$ with row $n_c+i$, and column $c_{i}$ with column $n_d+n'_r+i$.
        The result is shown in Fig. \ref{fig:diag_exten} with $l=l_1$, $\bR=\mathbf{0}$ and $n'_r$ reference variables.
        \begin{figure}
          \centering
          \hspace{-4em}
          \includegraphics[width=0.75\linewidth]{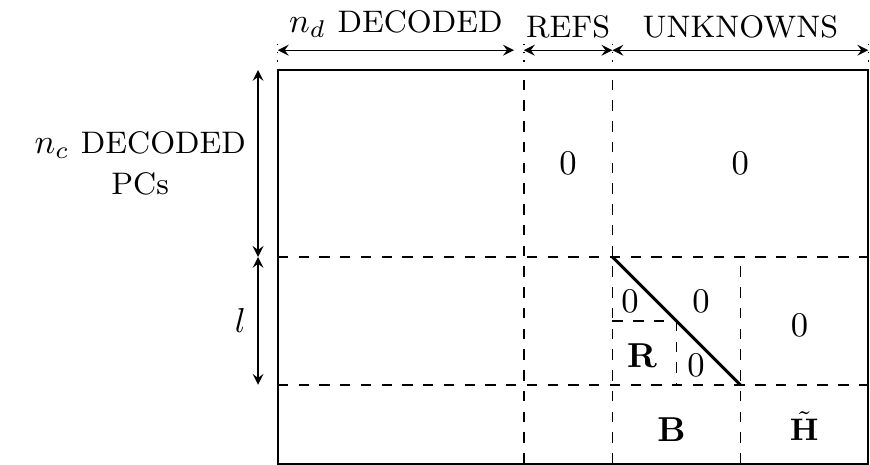}
          \caption{Diagonal Extension. The bold diagonal is filled with ones.}
          \label{fig:diag_exten}
        \end{figure}
        We proceed by applying additional diagonal extension steps on the rows of the obtained matrix below the $n_c+l_1$ first ones ($\tilde{\bH}$ in Fig. \ref{fig:diag_exten} with $l=l_1$) repeatedly, until this is not possible anymore (no more rows with a single unknown variable node that has not been diagonalized yet). By doing so we extend the diagonal of size $l=l_1$ that we have already constructed. Denote the total number of found rows with a single unknown variable node by $l_2 \ge l_1$. For example, the resulting PCM for two diagonal extension steps is also shown in Fig. \ref{fig:diag_exten} but now with $l=l_2$ and some $(l_2-l_1) \times l_1$ matrix $\bR$.
        Essentially, after we have chosen $n'_{r}$ reference variable nodes, we have applied a BP decoding procedure that also incorporates row and column permutations until convergence of the BP decoder in order to obtain the matrix shown in Fig. \ref{fig:diag_exten} with the largest possible $l=l_2$. We now repeat this process by an iterative application of the following basic procedure:
        \begin{enumerate}[wide, labelwidth=!, labelindent=0pt]
          \item Choose $n'_r$ additional reference variables from the unknwon variables that have not been diagonalized yet.
          \item Permute the PCM columns so that the columns corresponding to these additional reference variables are placed just after the columns of the reference variables from previous applications of the procedure.
          \item Apply as many diagonal extension steps as possible.
        \end{enumerate}

  \item \textbf{Expressing unknown variables as a combination of reference variables}. When we enter this stage the PCM has the form shown in Fig. \ref{fig:FinalPCM}. It contains $n_d$ BP decoded variables with known value, $\bd$, $n_r$ reference variables and $n_u$ remaining unknown variables. Both the values of the reference variables, $\br$, and the unknown variables, $\bu$, are unknown at this point. In this stage we express $\bu$ as an affine transformation of $\br$ (over GF(2)),
        \begin{equation}
          \bu=\bA\br + \ba
          \label{eq:u_Ar_b}
        \end{equation}
        where $\bA$ is a matrix of size $n_u \times n_r$, and $\ba$ is a vector of length $n_u$. Due to the triangulation and the sparsity of the PCM, this step can be computed efficiently using back-substitution as follows. Using the matrix form in Fig. \ref{fig:FinalPCM}, we have
        \begin{equation}
          \left( \begin{matrix}
              \bs^{(1)} \\ \bs^{(2)}
            \end{matrix} \right)
          =
          \left( \begin{matrix}
              \bH^{(1,2)} & \bH^{(1,3)} \\
              \bH^{(2,2)} & \bH^{(2,3)}
            \end{matrix} \right)
          \left( \begin{matrix}
              \br \\ \bu
            \end{matrix} \right)
          \label{eq:mat_eq_s1_s2}
        \end{equation}
        where $\bs^{(1)} = \bH^{(1,1)} \bd$ and $\bs^{(2)} = \bH^{(2,1)} \bd$.
        Since $\bH^{(1,3)}$ is a lower triangular, $n_u \times n_u$ matrix, with ones on its diagonal, we thus have for $l=1,2,...,n_u$,
        \begin{equation}
          u_l = s_l^{(1)} + \sum_j H_{l,j}^{(1,2)} r_j + \sum_{j<l} H_{l,j}^{(1,3)} u_j
          \label{eq:u_l}
        \end{equation}
        Suppose that we have already expressed $u_i$ as
        \begin{equation}
          u_i = \sum_j A_{i,j} r_j + a_i
          \label{eq:u_i}
        \end{equation}
        for $i=1,\ldots,k$, and wish to obtain a similar relation for $i=k+1$. Then, by \eqref{eq:u_l},
        \begin{equation}
          u_{k+1} = s_{k+1}^{(1)} + \sum_j H_{k+1,j}^{(1,2)} r_j + \sum_{i\in\cC_k} u_i
        \end{equation}
        where $\cC_k \defined \left\{ i \: : \: i\le k, H_{k+1,i}^{(1,3)}=1 \right\}$. Substituting \eqref{eq:u_i} for $u_i$ in the last summation and rearranging terms yields,
        \begin{equation}
          u_{k+1} = \sum_j A_{k+1,j} r_j + a_{k+1}
          \label{eq:u_k1}
        \end{equation}
        where
        \begin{equation}
          A_{k+1,j} = H_{k+1,j}^{(1,2)} + \sum_{i\in\cC_k} A_{i,j} , \quad
          a_{k+1} = s_{k+1}^{(1)} + \sum_{i\in\cC_k} a_{i}
        \end{equation}
        This shows that the rows of $\bA$ as well as the elements of the vector $\ba=(a_1,\ldots,a_{n_u})^T$ can be constructed recursively for $k=1,2,..$ as follows. Denote the $k$'th row of $\bA$ by $\overline{\ba}_k$ and the $k$'th row of $\bH^{(1,2)}$ by $\overline{{\bhh}}^{(1,2)}_k$. Then we first initialize by
        \begin{equation}
          \overline{\ba}_1 = \overline{\bhh}_1^{(1,2)}, \quad a_1 = s_1^{(1)} \label{eq:a_init}
        \end{equation}
        Then, for $k=1,\ldots,n_u-1$,
        \begin{equation}
          \overline{\ba}_{k+1} = \overline{\bhh}_{k+1}^{(1,2)} + \sum_{i\in\cC_k} \overline{\ba}_{i}
          , \quad
          a_{k+1} = s_{k+1}^{(1)} + \sum_{i\in\cC_k} a_{i} \label{eq:a_recursion}
        \end{equation}

  \item \textbf{Finding the values of the reference and unknown variables}. By \eqref{eq:mat_eq_s1_s2} and \eqref{eq:u_Ar_b},
        $$
          \bs^{(2)} = \bH^{(2,2)} \br + \bH^{(2,3)} \bu =
          \left( \bH^{(2,2)} + \bH^{(2,3)} \bA \right)\br + \bH^{(2,3)} \ba
        $$
        where, as was indicated above, $\bs^{(2)} = \bH^{(2,1)} \bd$. Hence, $\br$ is obtained by solving the following linear equation (over GF(2)),
        \begin{equation}
          \left( \bH^{(2,2)} + \bH^{(2,3)} \bA \right)\br = \bs^{(2)} + \bH^{(2,3)} \ba
          \label{eq:linear_system_Ar_b}
        \end{equation}
        To solve this system we apply Gaussian elimination on the $n_e \times (n_r+1)$ augmented matrix of \eqref{eq:linear_system_Ar_b}.
        This system must have a valid solution (the true transmitted codeword). Decoding will be successful if the true transmitted codeword is the unique solution.
        After we have obtained $\br$, we can also obtain $\bu$ from \eqref{eq:u_Ar_b}. Thus we have obtained the decoded codeword.
\end{enumerate}

\subsection{Complexity} \label{sec:EffMLComp}
Recall that the pruned PCM is obtained offline.
We now analyze the complexity of each stage of the efficient ML decoding algorithm described above in terms of the number of XORs (additions over GF(2)).
\begin{enumerate}[wide, labelwidth=!, labelindent=0pt,label=\textbf{\arabic*})]
  \item
        The complexity of BP decoding over the BEC is determined by the number of edges in the Tanner graph \cite{modern_coding_theory}. We start with a PCM with $\mathcal{O}(N\log_2 N)$ edges (corresponding to the FG in Fig. \ref{fig:PolarFG}). Then, after applying the pruning, the total number of edges decreases. Hence, the complexity of BP decoding over the pruned PCM is $\mathcal{O}(N\log_2 N)$.
  \item A diagonal extension step is equivalent to a BP iteration over the BEC without XORs.
        The total number of permutations is $\mathcal{O}(N\log_2 N)$.
  \item The number of XORs required to compute $\bs^{(1)} = \bH^{(1,1)} \bd$ is the number of ones in $\bH^{(1,1)}$. The recursion described by \eqref{eq:a_init} and \eqref{eq:a_recursion} requires $(n_r+1)(\gamma-n_u)$ XORs where $n_r$ is the number of reference variables and $\gamma$ is the number of ones in $\bH^{(1,3)}$. Denote by $d_c^{(1)}$ the average number of ones in a row of the final PCM in Fig. \ref{fig:FinalPCM} corresponding to the $\bH^{(1,\cdot)}$ matrices. Since the final PCM is sparse (it was obtained from the pruned PCM by row and column permutations only), $d_c^{(1)}$ is small. The total computational cost of this stage is $\mathcal{O}(d_c^{(1)} \cdot (n_r+1) \cdot n_u)$. We can also say that the total computational cost of this stage is $\mathcal{O}(d_c^{(1)} \cdot (n_r+1) \cdot N \log_2 N)$ (since $n_u=\mathcal{O}(N \log_2 N)$).
  \item The number of XORs required to compute $\bs^{(2)} = \bH^{(2,1)} \bd$ is the number of ones in $\bH^{(2,1)}$. We then compute $\bH^{(2,2)} + \bH^{(2,3)} \bA$ and $\bs^{(2)} + \bH^{(2,3)} \ba$, required in \eqref{eq:linear_system_Ar_b}, the complexity is $\mathcal{O}(\rho (n_r+1))$, where $\rho$ is the number of ones in $\bH^{(2,3)}$. Let $\rho = d_c^{(2)} n_e$, where $n_e$ is the number of remaining PCs and $d_c^{(2)}$ is the average number of ones in the $n_e$ remaining PC rows of the final PCM. Hence, the above complexity is $\mathcal{O}(d_c^{(2)} \cdot (n_r+1) \cdot n_e)$.
        Finally, the complexity of the Gaussian elimination required to solve the linear system \eqref{eq:linear_system_Ar_b} is $\mathcal{O}(n_e\cdot n_r^2)$.
        After we have obtained $\br$, we use \eqref{eq:u_Ar_b} to obtain $\bu$, the complexity of this additional step is $\mathcal{O}(n_u\cdot n_r)$.
\end{enumerate}

As can be seen, the complexity of the algorithm will be strongly influenced by $n_r$ and $n_e$.
We argue the following.
\begin{proposition} \label{prop1}
  When the code rate is below channel capacity, i.e., $R<C(\epsilon)=1-\epsilon$, we have,
  for any $q, p > 0$,
  \bre
  \lim_{N\rightarrow \infty} {\rm E} \{\left(N \log N\right)n_r^{q} \cdot n_e^{p}\} = 0
  \label{eq:prop_expected_nr}
  \ere
\end{proposition}
\begin{proof}
  By \cite[Lemma 6]{PCForChannelSrc}, when decoding a polar code transmitted over the BEC, BP on the standard polar FG cannot perform worse than SC. In addition, applying the BP decoder on the standard polar FG is equivalent to its application on the pruned graph (although the pruning steps have changed the message passing schedule, in the BEC case, the result of BP is invariant to the scheduling used). Now, the block error probability when using SC decoding is bounded by $P_e \leq 2^{-N^{\beta}}$, for any $\beta<1/2$, \cite{arikan2009on}. Hence, the error probability of the BP algorithm applied in the first stage of our proposed decoder is also bounded by the same term. Now, whenever BP fails to decode all bits, $n_r,n_e$ are bounded by $N \log N$. Thus, ${\rm E} \{\left(N \log N\right)n_r^q \cdot n_e^p\} \leq \left(N \log N\right)^{1 + p + q}\cdot 2^{-N^{\beta}}$. This immediately yields \eqref{eq:prop_expected_nr}.
\end{proof}

Now, the PCM pruning algorithm \cite{SparseGraphsBPPolar} can be modified such that the maximum degree of each PC node is at most some constant $d$. This modification is not required in practice, but with this modification, Proposition \ref{prop1} implies that the average computational complexity of the algorithm is $\mathcal{O}(N\log N)$. 
This complexity should be contrasted with that of straight-forward ML decoding \eqref{eq:MLDecBEC} which is $\mathcal{O}(N^3)$ ($N \cdot\epsilon$ variables where $\epsilon$ is the erasure probability and $N-K$ equations).

\subsection{CRC-Polar concatenation} \label{sec:AddCRC}
One approach for incorporating CRC in the proposed algorithm is to add the additional PC constraints associated with the CRC to the polar PCM before applying the pruning procedure. Unfortunately, the PCM of the CRC is not sparse, and this degrades the effectiveness of the pruning and results in a larger matrix.
As an example, for $\cP\left(256, 134\right)$ ($\cP\left(512, 262\right)$, respectively) with CRC of length $6$, the pruned PCM has blocklengh $N'=533$ ($N'=1150$), compared to $N'=355$ ($N'=773$) for a plain polar PCM.

Therefore, we used an alternative approach, where we add the additional CRC constraints after the pruning process. Since the CRC is a high rate code, we only add a small number of equations.
First, we need to obtain the CRC constraints in terms of the codeword (last $N$ columns of the pruned polar PCM).
Denote by $\bH_{CRC}$ the PCM of the CRC, i.e., for every information word $\bu$ (including the frozen bits), we have $\bH_{CRC}\cdot \bu = \mathbf{0}$. It is known from \cite{LP_Polar_decoding} that $\bu^T=\bc^T\cdot \bG_N$ where $\bG_N$ is the $N \times N$ generator matrix of the polar code. Hence, the CRC constraints can be expressed in terms of the codeword as $\mathbf{0}=\bH_{CRC}\cdot \bu = \bH_{CRC}\cdot \bG_N^T\cdot \bc$. That is, the CRC constraints that we add to the pruned polar PCM constraints are $\bH_{CRC}\bG_N^T$.
In order to decrease the density (number of ones) of $\bH_{CRC}\bG_N^T$, we used the greedy algorithm proposed in \cite{PolarBPCRCBrink}. This algorithm considers the Hamming weight of the sum of every two rows, $i$ and $j$. If this sum is smaller than that of either row $i$ or row $j$, then the algorithm replaces the row (either $i$ or $j$) with larger Hamming weight, with this sum.
\section{Simulation results}\label{sec:sim_res}
We used the standard concatenated CRC-Polar scheme with a total rate of $1/2$ and CRC of length $6$. Fig. \ref{fig:BEC_res} shows the BER / FER performance and the mean number of reference variables, $n_r$, and remaining equations, $n_e$ (when the codeword was successfully decoded by the BP algorithm, which is applied in the first stage, $n_r=n_e=0$). Whenever the diagonal cannot be further extended, we choose a single reference variable, that is $n'_r=1$. We considered two methods for selecting the reference variable. The first chooses the reference variable at random from all remaining unknown CVNs. The second, following \cite[Method C]{LDPC_ML}, chooses the unknown variable from a PC with the smallest number of remaining unknown variables. The second approach was slightly better and is hence the one used in Fig. \ref{fig:BEC_res}.
Since our decoder is the ML decoder, it achieves the best FER performance compared to all other decoders. As can be seen in Fig. \ref{fig:BEC_res}, our ML decoder also achieves the best BER performance. We can also see that the number of required reference variables, $n_r$, and the number of remaining equations, $n_e$, are small. For example, for $N=512$ and $\epsilon \le 0.37$ the average number of reference variables is less than $0.1\%$ of the code length.
\begin{figure}
  \centering
  \begin{subfigure}[b]{.49\linewidth}
    \includegraphics[width=\linewidth]{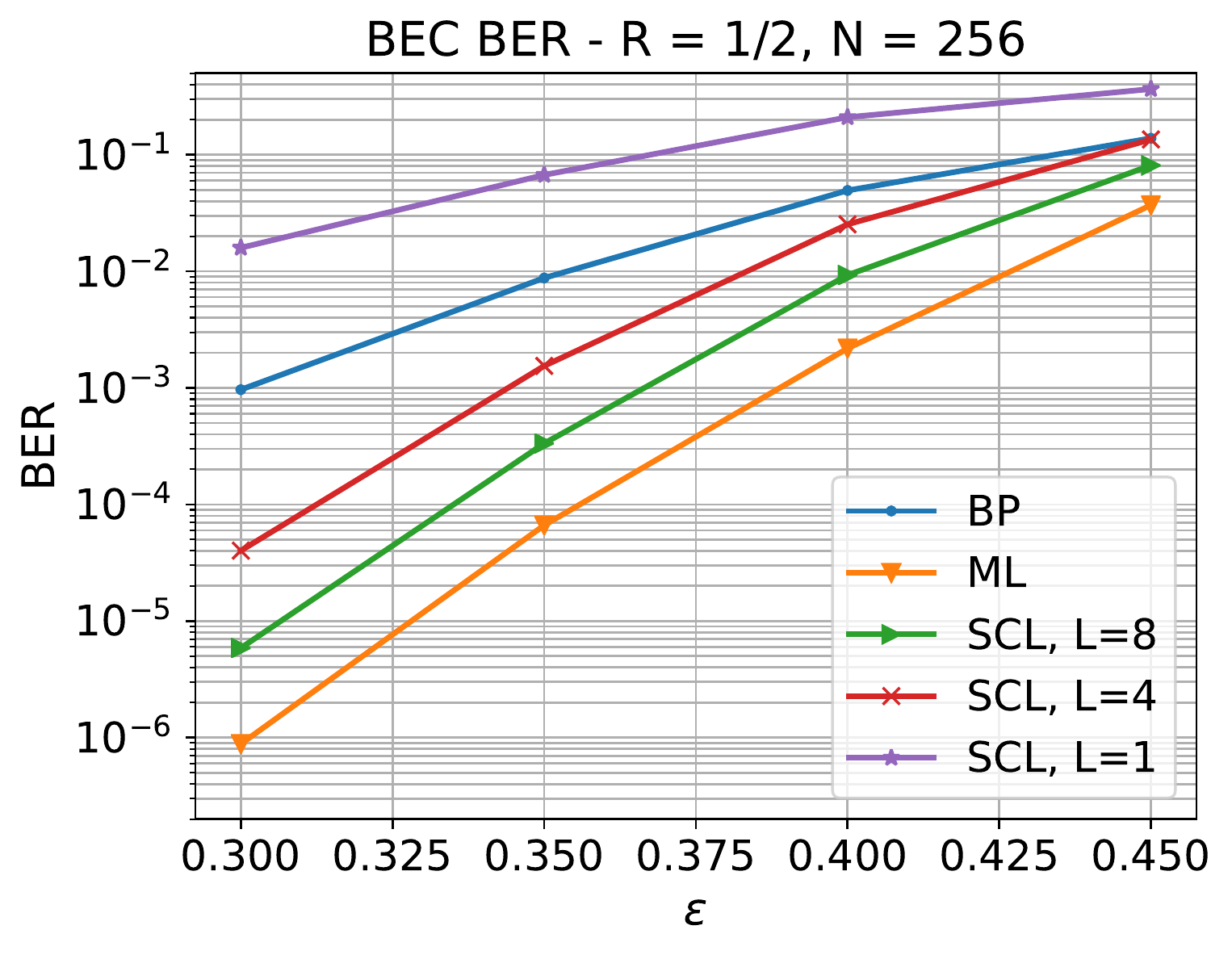}
  \end{subfigure}
  \begin{subfigure}[b]{.49\linewidth}
    \includegraphics[width=\linewidth]{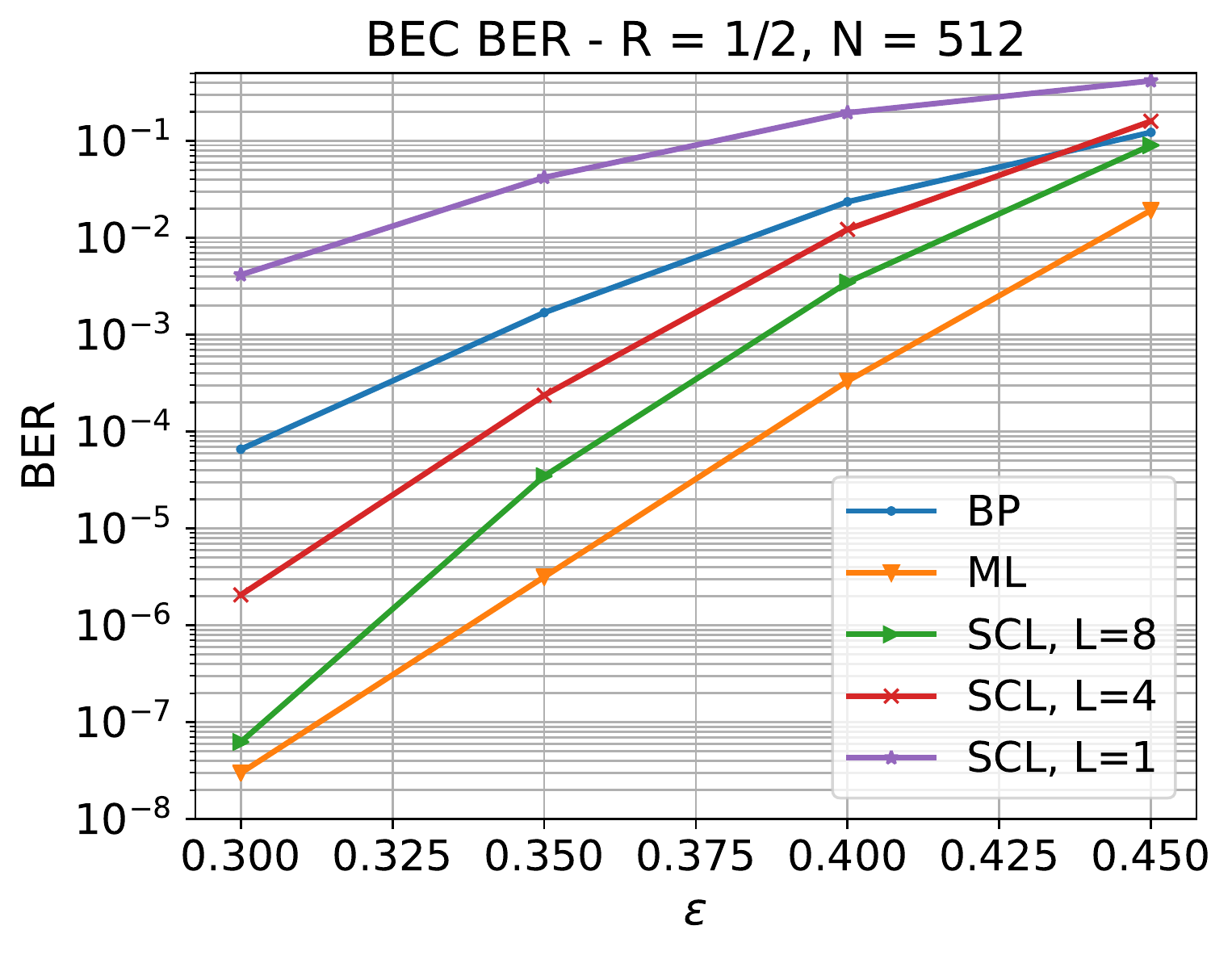}
  \end{subfigure}

  \begin{subfigure}[b]{.49\linewidth}
    \includegraphics[width=\linewidth]{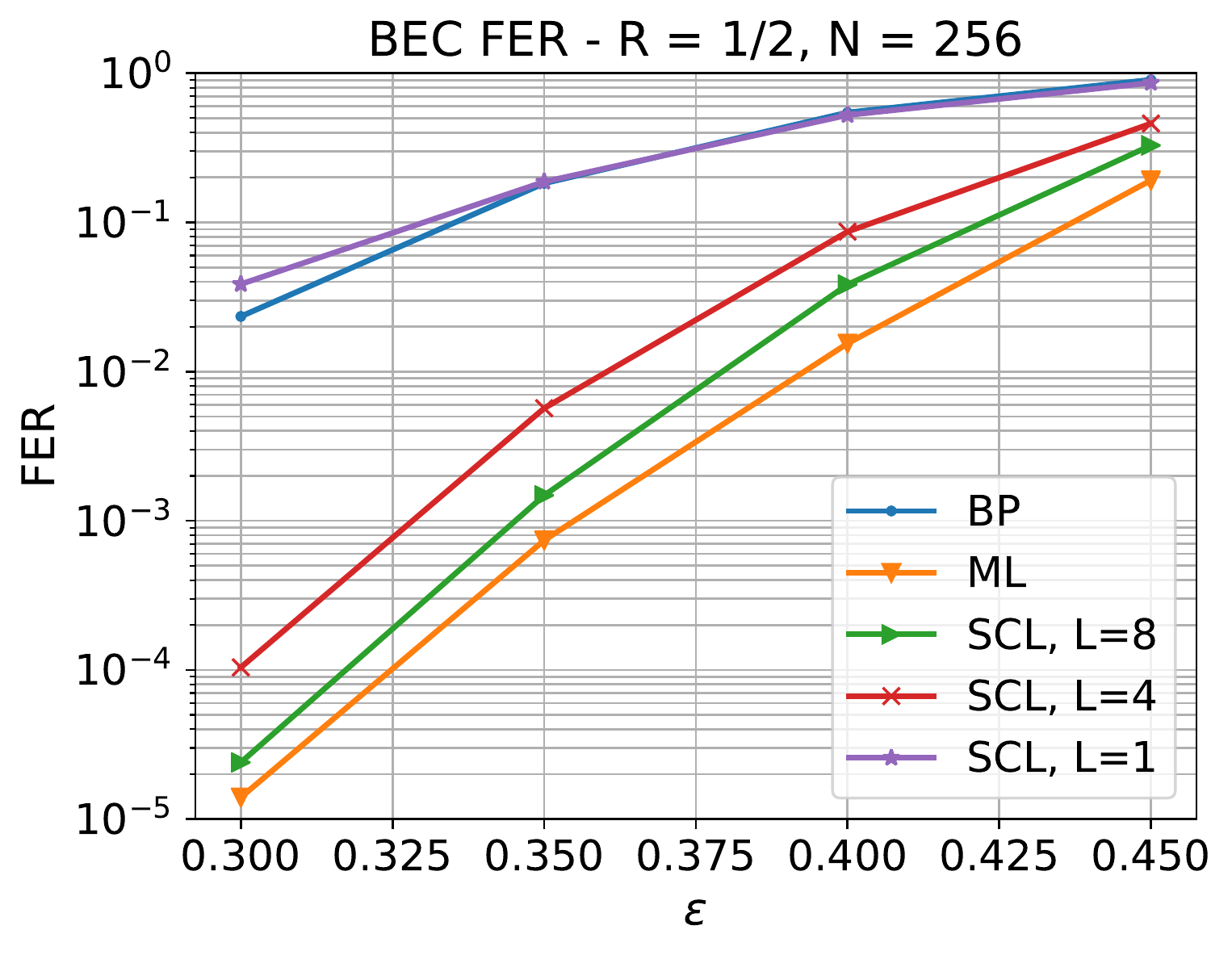}
  \end{subfigure}
  \begin{subfigure}[b]{.49\linewidth}
    \includegraphics[width=\linewidth]{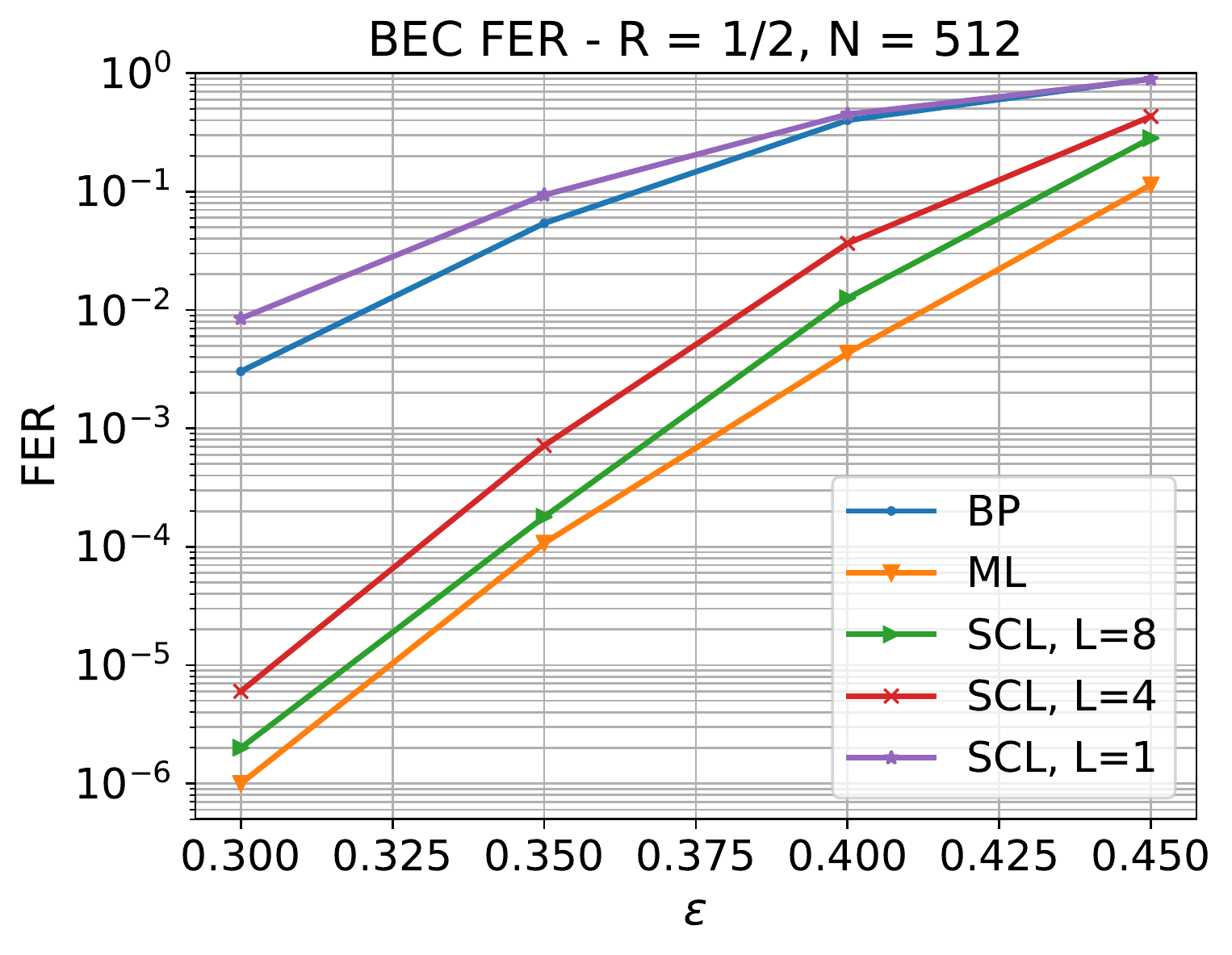}
  \end{subfigure}

  \begin{subfigure}[b]{.49\linewidth}
    \includegraphics[width=\linewidth]{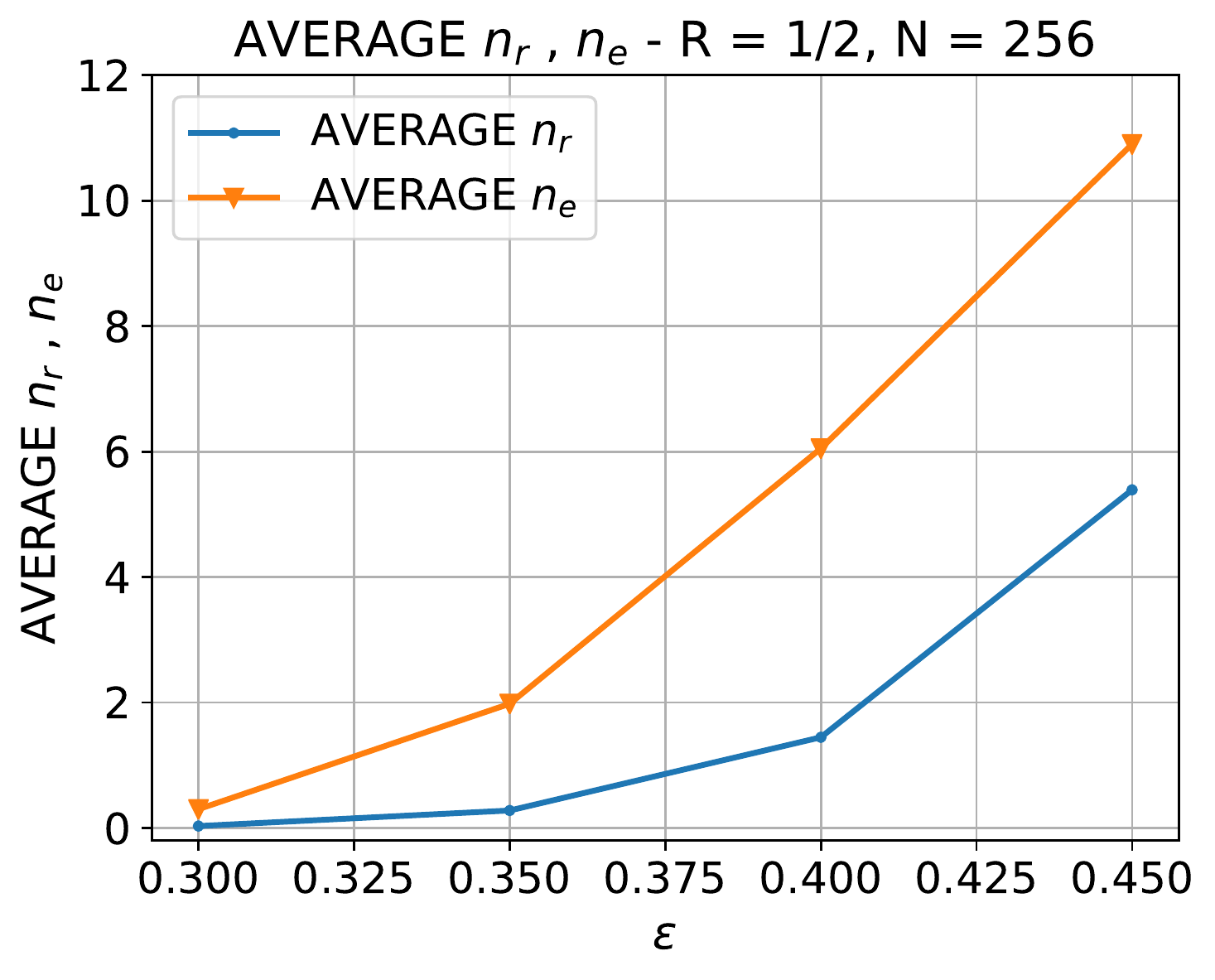}
  \end{subfigure}
  \begin{subfigure}[b]{.49\linewidth}
    \includegraphics[width=\linewidth]{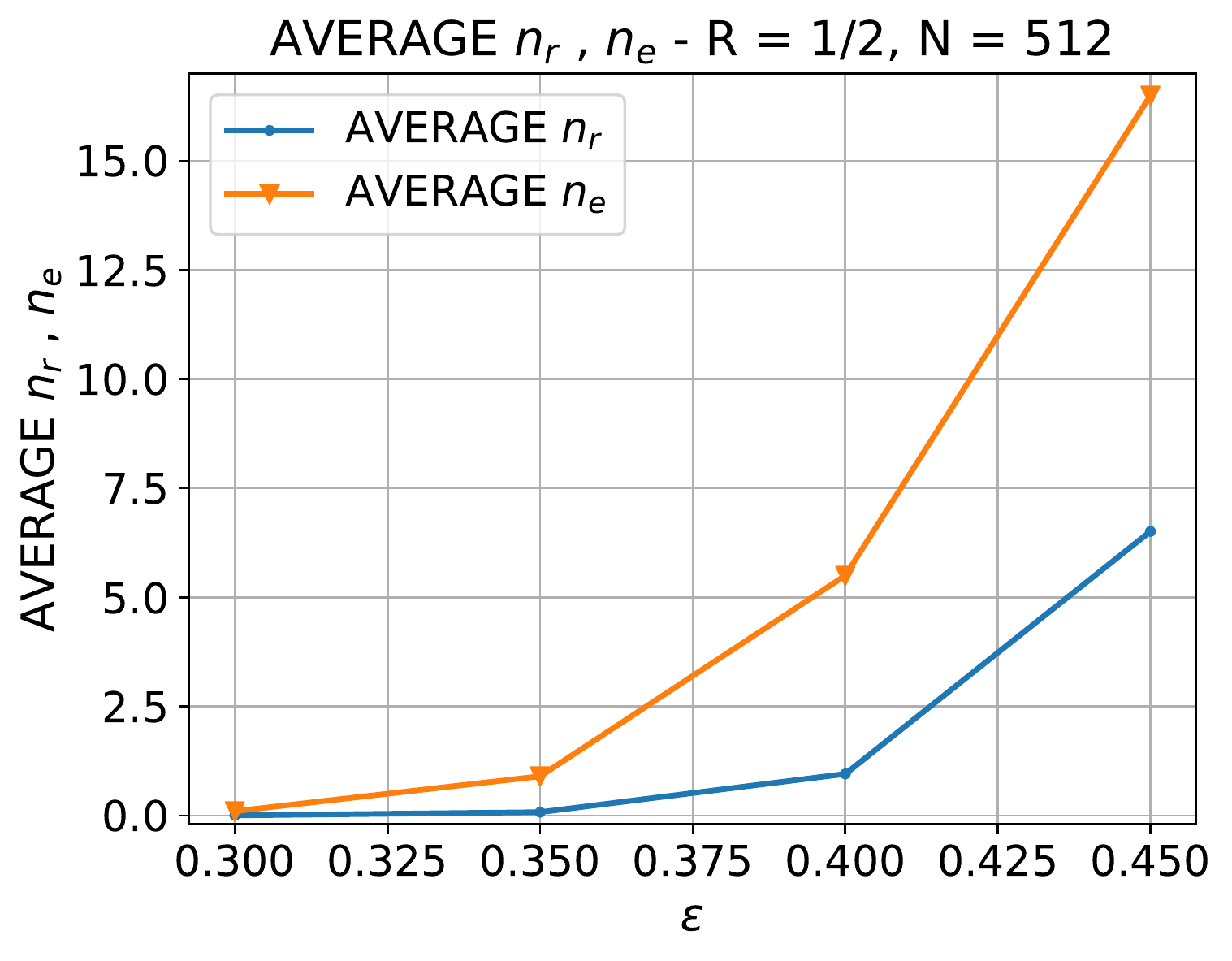}
  \end{subfigure}
  \caption{BER, FER, and average $n_r$ and $n_e$, for CRC concatenated polar codes with rate $1/2$ and different blocklengths.
  }
  \label{fig:BEC_res}
\end{figure}

\section{Parallel Implementation} \label{sec:parImp}
For efficient parallel implementation, we suggest modifying stage 2 in the decoding algorithm in Section \ref{sec:EffML}. This change simplifies stages 3 and 4.
Recall that the goal of stage 2 is to obtain the PCM shown in Fig. \ref{fig:FinalPCM}. The revised algorithm further specifies $\bH^{(1,3)}=I$ ($n_u \times n_u$ diagonal matrix) and $\bH^{(2,3)}=0$.
For that purpose we modify stage 2 by adding the following row elimination step after each diagonal extension step.
For example, suppose that the first diagonal extension step found $l_1$ rows with a single unknown variable node, so that after the first diagonal extension step, the PCM is as shown in Fig. \ref{fig:diag_exten} with $l=l_1$, $\bR=\mathbf{0}$ and $n'_r$ reference variables. At this point, we add a row elimination step: We use row additions (XORs) with the $l=l_1$ rows of the diagonal, in order to zero out the sub-matrix $\bB$ in Fig. \ref{fig:diag_exten} (i.e., after applying these additions, $\bB=\mathbf{0}$).
The point is that these row additions can be executed in parallel.
The same row elimination step is added after each diagonal extension step.

Stages 3 and 4 simplify considerably due to the modification in stage 2.
Using $\bH^{(1,3)} = I$ in \eqref{eq:mat_eq_s1_s2} yields \eqref{eq:u_Ar_b} for $\bA=\bH^{(1,2)}$, $\ba=\bH^{(1,1)}\bd$.
Using $\bH^{(2,3)}=0$ simplifies \eqref{eq:linear_system_Ar_b} to
$\bH^{(2,2)} \br = \bs^{(2)}$.
To solve this system, which is typically small ($n_e \times n_r$), we can use the parallel Gaussian elimination over GF(2) algorithm proposed in \cite{FastGausElim}.

Although $n'_r=1$ minimizes the total number of reference variables, $n_r$, used, for efficient parallel implementation of stage 2 it is beneficial to use $n'_r>1$.

\section*{Acknowledgment}
This research was supported by the Israel Science Foundation (grant no. 1868/18).

\clearpage
\IEEEtriggeratref{21}



\end{document}

%% file: shortcuts.tex
\newcommand{\rem}[1]{}

\newtheorem{proposition}{Proposition}

\newcommand{\bre}{\begin{equation}}
\newcommand{\ere}{\end{equation}}

\newcommand{\ee}\]
\newcommand{\bra}{\begin{eqnarray}}
\newcommand{\era}{\end{eqnarray}}
\newcommand{\bfg}{\begin{figure}[hbtp]}
\newcommand{\efg}{\end{figure}}
\newcommand{\bver}{\begin{verbatim}}
\newcommand{\ever}{\end{verbatim}}
\newcommand{\bit}{\begin{itemize}}
\newcommand{\eit}{\end{itemize}}
\newcommand{\ben}{\begin{enumerate}}
\newcommand{\een}{\end{enumerate}}

\newcommand{\bA}{{\bf A}}
\newcommand{\bB}{{\bf B}}

\newcommand{\bG}{{\bf{G}}}

\newcommand{\bF}{{\bf F}}

\newcommand{\bc}{{\bf c}}

\newcommand{\bR}{{\bf R}}

\newcommand{\baa}{\begin{eqnarray*}}
\newcommand{\eaa}{\end{eqnarray*}}

\newcommand{\bs}{{\bf s}}
\newcommand{\br}{{\bf r}}
\newcommand{\ba}{{\bf a}}

\newcommand{\bhh}{{\bf h}}

\newcommand{\bH}{{\bf H}}

\newcommand{\bu}{{\bf u}}

\newcommand{\bd}{{\bf d}}

\newcommand{\by}{{\bf y}}

\newcommand{\cP}{{\cal P}}

\newcommand{\cK}{{\cal K}}

\newcommand{\cC}{{\mathcal{C}}}

\newcommand{\defined}{\triangleq}

\def\defined{\: {\stackrel{\scriptscriptstyle \Delta}{=}} \: }

\newfont{\boldlarge}{msbm10 scaled 1100}

\newcommand{\comment}[1]{}

\newlength{\tmpbigbar}


%% file: Polar_ML_BEC.bbl
\begin{thebibliography}{10}
	\providecommand{\url}[1]{#1}
	\csname url@samestyle\endcsname
	\providecommand{\newblock}{\relax}
	\providecommand{\bibinfo}[2]{#2}
	\providecommand{\BIBentrySTDinterwordspacing}{\spaceskip=0pt\relax}
	\providecommand{\BIBentryALTinterwordstretchfactor}{4}
	\providecommand{\BIBentryALTinterwordspacing}{\spaceskip=\fontdimen2\font plus
		\BIBentryALTinterwordstretchfactor\fontdimen3\font minus
		\fontdimen4\font\relax}
	\providecommand{\BIBforeignlanguage}[2]{{%
			\expandafter\ifx\csname l@#1\endcsname\relax
			\typeout{** WARNING: IEEEtran.bst: No hyphenation pattern has been}%
			\typeout{** loaded for the language `#1'. Using the pattern for}%
			\typeout{** the default language instead.}%
			\else
			\language=\csname l@#1\endcsname
			\fi
			#2}}
	\providecommand{\BIBdecl}{\relax}
	\BIBdecl
	
	\bibitem{PolarCodes}
	E.~Arikan, ``Channel polarization: A method for constructing capacity-achieving
	codes for symmetric binary-input memoryless channels,'' \emph{IEEE
		Transactions on Information Theory}, vol.~55, no.~7, pp. 3051--3073, Jul
	2009.
	
	\bibitem{SCL}
	I.~{Tal} and A.~{Vardy}, ``List decoding of polar codes,'' \emph{IEEE
		Transactions on Information Theory}, vol.~61, no.~5, pp. 2213--2226, 2015.
	
	\bibitem{alamdar2011simplified}
	A.~Alamdar-Yazdi and F.~R. Kschischang, ``A simplified successive-cancellation
	decoder for polar codes,'' \emph{IEEE communications letters}, vol.~15,
	no.~12, pp. 1378--1380, 2011.
	
	\bibitem{leroux2013semi}
	C.~Leroux, A.~J. Raymond, G.~Sarkis, and W.~J. Gross, ``{A semi-parallel
		successive-cancellation decoder for polar codes},'' \emph{IEEE Transactions
		on Signal Processing}, vol.~61, no.~2, pp. 289--299, 2013.
	
	\bibitem{sarkis2014fast}
	G.~Sarkis, P.~Giard, A.~Vardy, C.~Thibeault, and W.~J. Gross, ``Fast polar
	decoders: Algorithm and implementation,'' \emph{IEEE Journal on Selected
		Areas in Communications}, vol.~32, no.~5, pp. 946--957, 2014.
	
	\bibitem{li2014low}
	B.~Li, H.~Shen, D.~Tse, and W.~Tong, ``{Low-latency polar codes via hybrid
		decoding},'' in \emph{Proc. 8th Int. Symp. Turbo Codes and Iterative Inf.
		Processing (ISTC)}, August 2014, pp. 223--227.
	
	\bibitem{balatsoukas2015llr}
	A.~Balatsoukas-Stimming, M.~B. Parizi, and A.~Burg, ``{LLR}-based successive
	cancellation list decoding of polar codes,'' \emph{IEEE transactions on
		signal processing}, vol.~63, no.~19, pp. 5165--5179, 2015.
	
	\bibitem{yuan2015low}
	B.~Yuan and K.~Parhi, ``{Low-latency successive-cancellation list decoders for
		polar codes with multibit decision},'' \emph{IEEE Trans. Very Large Scale
		Integr. (VLSI) Syst.}, vol.~23, no.~10, pp. 2268--2280, 2015.
	
	\bibitem{xiong2015symbol}
	C.~Xiong, J.~Lin, and Z.~Yan, ``{Symbol-decision successive cancellation list
		decoder for polar codes},'' \emph{IEEE Transactions on Signal Processing},
	vol.~64, no.~3, pp. 675--687, February 2016.
	
	\bibitem{chen2016reduce}
	K.~Chen, B.~Li, H.~Shen, J.~Jin, and D.~Tse, ``Reduce the complexity of list
	decoding of polar codes by tree-pruning,'' \emph{IEEE Communications
		Letters}, vol.~20, no.~2, pp. 204--207, 2016.
	
	\bibitem{hashemi2018decoder}
	S.~A. Hashemi, M.~Mondelli, S.~H. Hassani, C.~Condo, R.~L. Urbanke, and W.~J.
	Gross, ``Decoder partitioning: Towards practical list decoding of polar
	codes,'' \emph{IEEE Transactions on Communications}, vol.~66, no.~9, pp.
	3749--3759, 2018.
	
	\bibitem{hashemi2018decoding}
	S.~A. Hashemi, N.~Doan, M.~Mondelli, and W.~J. Gross, ``Decoding reed-muller
	and polar codes by successive factor graph permutations,'' in \emph{2018 IEEE
		10th International Symposium on Turbo Codes \& Iterative Information
		Processing (ISTC)}.\hskip 1em plus 0.5em minus 0.4em\relax IEEE, 2018, pp.
	1--5.
	
	\bibitem{giard2018fast}
	P.~Giard and A.~Burg, ``Fast-{SSC}-flip decoding of polar codes,'' in
	\emph{2018 IEEE Wireless Communications and Networking Conference Workshops
		(WCNCW)}.\hskip 1em plus 0.5em minus 0.4em\relax IEEE, 2018, pp. 73--77.
	
	\bibitem{hashemi2019rate}
	S.~A. Hashemi, C.~Condo, M.~Mondelli, and W.~J. Gross, ``Rate-flexible fast
	polar decoders,'' \emph{IEEE Transactions on Signal Processing}, vol.~67,
	no.~22, pp. 5689--5701, 2019.
	
	\bibitem{Arkan2010PolarC}
	E.~Arikan, ``Polar codes : A pipelined implementation,'' in \emph{Proc. 4th
		Int. Symp. on Broad. Commun. (ISBC)}, 2010, pp. 11--14.
	
	\bibitem{polar_vs_reed}
	E.~{Arikan}, ``{A performance comparison of polar codes and Reed-Muller
		codes},'' \emph{IEEE Communications Letters}, vol.~12, no.~6, pp. 447--449,
	2008.
	
	\bibitem{eslami2010on}
	A.~Eslami and H.~Pishro-Nik, ``On bit error rate performance of polar codes in
	finite regime,'' in \emph{2010 48th Annual Allerton Conference on
		Communication, Control, and Computing (Allerton)}, 2010, pp. 188--194.
	
	\bibitem{BP_arc}
	B.~{Yuan} and K.~K. {Parhi}, ``{Architecture optimizations for BP polar
		decoders},'' in \emph{2013 IEEE International Conference on Acoustics, Speech
		and Signal Processing}, 2013, pp. 2654--2658.
	
	\bibitem{BP_BEC}
	A.~{Eslami} and H.~{Pishro-Nik}, ``On finite-length performance of polar codes:
	Stopping sets, error floor, and concatenated design,'' \emph{IEEE
		Transactions on Communications}, vol.~61, no.~3, pp. 919--929, 2013.
	
	\bibitem{Polar_LDPC_conc}
	J.~{Guo}, M.~{Qin}, A.~{Guillen i Fabregas}, and P.~H. {Siegel}, ``Enhanced
	belief propagation decoding of polar codes through concatenation,'' in
	\emph{2014 IEEE International Symposium on Information Theory}, 2014, pp.
	2987--2991.
	
	\bibitem{bp_early_term}
	B.~{Yuan} and K.~K. {Parhi}, ``Early stopping criteria for energy-efficient
	low-latency belief-propagation polar code decoders,'' \emph{IEEE Transactions
		on Signal Processing}, vol.~62, no.~24, pp. 6496--6506, 2014.
	
	\bibitem{crc_early_term}
	Y.~{Ren}, C.~{Zhang}, X.~{Liu}, and X.~{You}, ``Efficient early termination
	schemes for belief-propagation decoding of polar codes,'' in \emph{2015 IEEE
		11th International Conference on ASIC (ASICON)}, 2015, pp. 1--4.
	
	\bibitem{PCForChannelSrc}
	N.~Hussami, S.~B. Korada, and R.~Urbanke, ``Performance of polar codes for
	channel and source coding,'' in \emph{IEEE International Symposium on
		Information Theory (ISIT)}, 2009, pp. 1488--1492.
	
	\bibitem{PolarBPCRCWarren}
	N.~Doan, S.~A. Hashemi, E.~N. Mambou, T.~Tonnellier, and W.~J. Gross, ``Neural
	belief propagation decoding of {CRC}-polar concatenated codes,'' in
	\emph{IEEE International Conference on Communications (ICC)}, 2019, pp. 1--6.
	
	\bibitem{BPPermuted}
	A.~Elkelesh, M.~Ebada, S.~Cammerer, and S.~ten Brink, ``Belief propagation
	decoding of polar codes on permuted factor graphs,'' in \emph{IEEE Wireless
		Communications and Networking Conference (WCNC)}, 2018, pp. 1--6.
	
	\bibitem{BPL}
	------, ``Belief propagation list decoding of polar codes,'' \emph{IEEE
		Communications Letters}, vol.~22, no.~8, pp. 1536--1539, Aug 2018.
	
	\bibitem{BPPermutedWarren}
	N.~Doan, S.~A. Hashemi, M.~Mondelli, and W.~J. Gross, ``On the decoding of
	polar codes on permuted factor graphs,'' in \emph{IEEE Global Communications
		Conference (GLOBECOM)}, 2018, pp. 1--6.
	
	\bibitem{yu2019belief}
	Y.~Yu, Z.~Pan, N.~Liu, and X.~You, ``Belief propagation bit-flip decoder for
	polar codes,'' \emph{IEEE Access}, vol.~7, pp. 10\,937--10\,946, 2019.
	
	\bibitem{PolarBPCRCBrink}
	M.~Geiselhart, A.~Elkelesh, M.~Ebada, S.~Cammerer, and S.~ten Brink,
	``{CRC}-aided belief propagation list decoding of polar codes,'' \emph{arXiv
		preprint arXiv:2001.05303}, 2020.
	
	\bibitem{pishro2004on}
	H.~Pishro-Nik and F.~Fekri, ``On decoding of low-density parity-check codes
	over the binary erasure channel,'' \emph{IEEE Transactions on Information
		Theory}, vol.~50, no.~3, pp. 439--454, 2004.
	
	\bibitem{LDPC_ML}
	D.~{Burshtein} and G.~{Miller}, ``Efficient maximum-likelihood decoding of
	{LDPC} codes over the binary erasure channel,'' \emph{IEEE Transactions on
		Information Theory}, vol.~50, no.~11, pp. 2837--2844, 2004.
	
	\bibitem{shokrollahi2006systems}
	A.~Shokrollahi, S.~Lassen, and R.~Karp, ``Systems and processes for decoding
	chain reaction codes through inactivation,'' US patent 6,856,263, 2005.
	
	\bibitem{cocskun2020successive}
	M.~C. Co{\c{s}}kun, J.~Neu, and H.~D. Pfister, ``Successive cancellation
	inactivation decoding for modified {Reed-Muller} and {eBCH} codes,'' in
	\emph{2020 IEEE International Symposium on Information Theory (ISIT)}, 2020,
	pp. 437--442.
	
	\bibitem{algebraic_coding_theory}
	E.~R. Berlekamp, \emph{Algebraic coding theory}, ser. McGraw-Hill series in
	systems science.\hskip 1em plus 0.5em minus 0.4em\relax McGraw-Hill, 1968.
	
	\bibitem{SparseGraphsBPPolar}
	S.~Cammerer, M.~Ebada, A.~Elkelesh, and S.~ten Brink, ``Sparse graphs for
	belief propagation decoding of polar codes,'' in \emph{IEEE International
		Symposium on Information Theory (ISIT)}, 2018, pp. 1465--1469.
	
	\bibitem{LP_Polar_decoding}
	N.~{Goela}, S.~B. {Korada}, and M.~{Gastpar}, ``On {LP} decoding of polar
	codes,'' in \emph{2010 IEEE Information Theory Workshop}, 2010, pp. 1--5.
	
	\bibitem{cover_book}
	T.~M. Cover and J.~A. Thomas, \emph{Elements of Information Theory},
	2nd~ed.\hskip 1em plus 0.5em minus 0.4em\relax New York: Wiley, 2006.
	
	\bibitem{modern_coding_theory}
	T.~Richardson and R.~Urbanke, \emph{Modern Coding Theory}.\hskip 1em plus 0.5em
	minus 0.4em\relax Cambridge University Press, 2008.
	
	\bibitem{EffEncLDPC}
	T.~J. {Richardson} and R.~L. {Urbanke}, ``Efficient encoding of low-density
	parity-check codes,'' \emph{IEEE Transactions on Information Theory},
	vol.~47, no.~2, pp. 638--656, 2001.
	
	\bibitem{arikan2009on}
	E.~Arikan and E.~Telatar, ``On the rate of channel polarization,'' in
	\emph{2009 IEEE International Symposium on Information Theory}, 2009, pp.
	1493--1495.
	
	\bibitem{FastGausElim}
	A.~{Rupp}, J.~{Pelzl}, C.~{Paar}, M.~C. {Mertens}, A.~{Bogdanov}, A.~{Rupp},
	J.~{Pelzl}, C.~{Paar}, M.~C. {Mertens}, and A.~{Bogdanov}, ``A parallel
	hardware architecture for fast {Gaussian} elimination over {GF(2)},'' in
	\emph{2006 14th Annual IEEE Symposium on Field-Programmable Custom Computing
		Machines}, 2006, pp. 237--248.
	
\end{thebibliography}
